\documentclass[11pt]{article}

\usepackage[utf8]{inputenc}
\usepackage[T1]{fontenc}
\usepackage{amsmath,amsthm,amssymb}
\usepackage{hyperref}
\usepackage{algorithm}
\usepackage{algpseudocode}
\usepackage{booktabs}
\usepackage{geometry}
\usepackage{float}
\usepackage{microtype}

\hypersetup{
    pdftitle={Qatsi: Stateless Secret Generation via Hierarchical Memory-Hard Key Derivation},
    pdfauthor={René Coignard, Anton Rygin},
    pdfsubject={Cryptography, Key Derivation Functions, Password Management},
    pdfkeywords={Argon2id, memory-hard functions, hierarchical derivation, stateless authentication, key derivation},
    colorlinks=true,
    linkcolor=blue,
    urlcolor=blue,
    citecolor=blue,
}

\newtheorem{theorem}{Theorem}
\newtheorem{lemma}[theorem]{Lemma}
\newtheorem{definition}{Definition}
\newtheorem{proposition}[theorem]{Proposition}

\title{Qatsi: Stateless Secret Generation\\via Hierarchical Memory-Hard Key Derivation}

\author{
    \begin{tabular}{cc}
        \begin{tabular}{c}
            René Coignard\thanks{Email: \texttt{contact@renecoignard.com}. ORCID: 0009-0004-0484-3336}\\
            {\small Independent Researcher}\\
            {\small Dessau-Roßlau, Germany}
        \end{tabular}
        &
        \begin{tabular}{c}
            Anton Rygin\thanks{Email: \texttt{anton.rygin@tu-dresden.de}. ORCID: 0009-0009-3640-9872}\\
            {\small Technische Universität Dresden}\\
            {\small Berlin, Germany}
        \end{tabular}
    \end{tabular}\\[1cm]
    {\small \url{https://github.com/coignard/qatsi}}
}

\date{October 2025}

\begin{document}

\maketitle

\begin{abstract}
We present Qatsi, a hierarchical key derivation scheme using Argon2id that generates reproducible cryptographic secrets without persistent storage. The system eliminates vault-based attack surfaces by deriving all secrets deterministically from a single high-entropy master secret and contextual layers. Outputs achieve 103--312~bits of entropy through memory-hard derivation (64--128~MiB, 16--32~iterations) and provably uniform rejection sampling over 7776-word mnemonics or 90-character passwords. We formalize the hierarchical construction, prove output uniformity, and quantify GPU attack costs: $2.4 \times 10^{16}$~years for 80-bit master secrets on single-GPU adversaries under Paranoid parameters (128~MiB memory). The implementation in Rust provides automatic memory zeroization, compile-time wordlist integrity verification, and comprehensive test coverage. Reference benchmarks on Apple M1 Pro (2021) demonstrate practical usability with 544~ms Standard mode and 2273~ms Paranoid mode single-layer derivations. Qatsi targets air-gapped systems and master credential generation where stateless reproducibility outweighs rotation flexibility.
\end{abstract}


\section{Introduction}

Password management architectures face a fundamental trade-off between security and usability. Centralized encrypted vaults (KeePassXC, Bitwarden) enable arbitrary credential storage and cross-device synchronization but concentrate risk in a single encrypted database. Vault compromise, cloud interception, or file exfiltration exposes all credentials. Deterministic generation eliminates persistent attack surfaces but sacrifices per-credential rotation and key isolation.

Argon2id, winner of the 2015 Password Hashing Competition~\cite{phc-2015} and standardized in RFC~9106~\cite{rfc9106}, remains underutilized in consumer password management despite its proven resistance to GPU/ASIC parallelization through memory-hardness. Mainstream alternatives like PBKDF2 (RFC~2898~\cite{rfc2898}) lack memory-hard properties and remain vulnerable to hardware acceleration.

Existing deterministic schemes include Diceware~\cite{diceware} (physical dice, manual), BIP39~\cite{bip39} (PBKDF2-based, cryptocurrency wallets), and LessPass (stateless web passwords). None provide hierarchical context-aware derivation with memory-hard KDF and provably uniform output sampling.

\textbf{Problem.} Generate cryptographically strong, reproducible secrets for master credentials (password manager master passwords, full-disk encryption passphrases, PGP/SSH key passphrases) without storing any persistent state, while resisting GPU attacks through memory-hardness.

\textbf{Contributions.}
\begin{enumerate}
\item Hierarchical Argon2id chaining enabling context-specific derivation from a single master secret without storage overhead.
\item Formal proof of output uniformity via rejection sampling for both word-based (7776-word) and character-based (90-char) alphabets.
\item Quantitative GPU attack cost analysis with realistic hardware parameters showing $\approx 10^{16}$-year resistance for 80-bit entropy under 128~MiB memory constraints.
\item Production implementation in Rust with automatic zeroization, Unicode NFC normalization, and comprehensive testing including regression vectors verified on Apple M1 Pro hardware.
\end{enumerate}

\textbf{Non-goals.} Qatsi does not replace traditional password managers for everyday website credentials with varying policies, existing passwords, or frequent rotation requirements. It targets high-stakes reproducible secrets in constrained environments (air-gapped systems, master credentials).

\section{Related Work}

\subsection{Memory-Hard Key Derivation}

Argon2~\cite{rfc9106} achieves memory-hardness by filling memory with pseudorandom data dependent on password and salt, forcing adversaries to allocate comparable memory per guess. Argon2id combines data-independent (Argon2i) and data-dependent (Argon2d) phases, resisting both side-channel attacks and GPU optimization. Standard configurations use 64~MiB memory and 16 iterations; cryptographic libraries recommend $\geq 2048$ iterations for PBKDF2 to achieve comparable security~\cite{owasp-pbkdf2}.

scrypt~\cite{percival-scrypt} pioneered memory-hard password hashing but achieved less widespread standardization than Argon2. PBKDF2~\cite{rfc2898}, despite decades of deployment, remains vulnerable to GPU farms processing $>10^9$ hashes/second on consumer hardware~\cite{hashcat-bench}.

\textbf{Distinction from prior work.} While Argon2 provides single-level key derivation, Qatsi introduces hierarchical chaining with context-aware layers, enabling stateless multi-domain secret generation from a single root secret without storage overhead.

\subsection{Deterministic Password Schemes}

\textbf{Diceware}~\cite{diceware} generates passphrases by mapping physical dice rolls to wordlists. The EFF Large Wordlist~\cite{eff-wordlist} contains exactly $6^5 = 7776$ words, yielding $\log_2(7776) \approx 12.925$ bits per word. Diceware requires manual dice rolling; Qatsi automates generation via ChaCha20 keystream with provably uniform rejection sampling.

\textbf{BIP39}~\cite{bip39} encodes 128--256 bit seeds as 12--24 word mnemonics using a 2048-word list and PBKDF2-HMAC-SHA512 (2048 iterations). The scheme appends 4--8 checksum bits to the original entropy, requiring 12--24 words total for 128--256 bits of entropy. BIP39 optimizes for cryptocurrency wallet recovery; Qatsi generalizes to arbitrary hierarchical secrets with stronger memory-hard KDF.

\textbf{LessPass} implements stateless password generation via PBKDF2 but lacks hierarchical derivation and memory-hard protection.

\textbf{Our contribution.} Qatsi is the first system combining hierarchical context-aware derivation, memory-hard KDF (Argon2id), and provably uniform output sampling for general-purpose stateless secret management.

\subsection{Unbiased Sampling}

Lemire~\cite{lemire-fast-rejection} formalized fast unbiased random integer generation via rejection sampling. Modulo reduction introduces bias when range size does not divide sample space size. For uniform sampling from $[0, n)$ using $b$-bit integers: sample $r \in [0, 2^b)$; accept if $r < 2^b - (2^b \bmod n)$; otherwise reject. Expected samples per output: $(1 - \frac{2^b \bmod n}{2^b})^{-1}$.

Our implementation (Algorithms~\ref{alg:mnemonic}, \ref{alg:password}) applies this technique to word and character selection with formal uniformity proof (Theorem~\ref{thm:entropy}).

\section{Hierarchical Derivation}

\subsection{Definitions}

\begin{definition}[Hierarchical Key Derivation Function]
\label{def:hkdf}
Let $\mathcal{M} = \{0,1\}^*$ be the master secret space and $\mathcal{L} = (\{0,1\}^*)^n$ be an ordered sequence of $n$ context layers. Define the hierarchical key derivation function $\mathcal{H}: \mathcal{M} \times \mathcal{L} \to \{0,1\}^{256}$ as:

\begin{equation}
\mathcal{H}(M, (L_1, \ldots, L_n)) = K_n
\end{equation}

where $K_0 = M$ and for $i \in [1,n]$:
\begin{equation}
K_i = \textsc{Argon2id}(K_{i-1}, \textsc{Salt}(L_i), m, t, p, \ell)
\end{equation}

Parameters: $m$ (memory cost in KiB), $t$ (iterations), $p$ (parallelism), $\ell = 32$ bytes output length.
\end{definition}

\begin{definition}[Salt Preprocessing]
For input $L \in \{0,1\}^*$:
\begin{equation}
\textsc{Salt}(L) = 
\begin{cases}
L & \text{if } |L| \geq 16 \text{ bytes} \\
\textsc{Blake2b-512}(L) & \text{if } |L| < 16 \text{ bytes}
\end{cases}
\end{equation}
\end{definition}

Argon2 requires salts $\geq 8$ bytes~\cite{rfc9106}; we enforce 16 bytes minimum for additional security margin. Inputs shorter than 16 bytes are expanded to 64 bytes via BLAKE2b-512, ensuring sufficient entropy and preventing short-salt attacks while maintaining determinism.

\begin{definition}[Input Normalization]
All text inputs undergo:
\begin{equation}
\textsc{Normalize}(s) = \textsc{NFC}(\textsc{Trim}(s))
\end{equation}
where $\textsc{NFC}$ applies Unicode Normalization Form C and $\textsc{Trim}$ removes leading/trailing whitespace.
\end{definition}

\subsection{Configuration Profiles}

\begin{table}[H]
\centering
\begin{tabular}{@{}lrrrr@{}}
\toprule
\textbf{Profile} & $m$ (MiB) & $t$ & $p$ & \textbf{M1 Pro Time (ms)} \\
\midrule
Standard & 64 & 16 & 6 & 544 \\
Paranoid & 128 & 32 & 6 & 2273 \\
\bottomrule
\end{tabular}
\caption{Argon2id parameter configurations with measured single-layer derivation times on Apple M1 Pro (2021). High-end GPUs achieve estimated 0.25--1.0~s (Standard) and 0.6--1.25~s (Paranoid).}
\label{tab:profiles}
\end{table}

Output length fixed at $\ell = 32$ bytes for compatibility with ChaCha20~\cite{rfc8439}.

\subsection{Output Generation}

\begin{algorithm}
\caption{Mnemonic generation via rejection sampling}
\label{alg:mnemonic}
\begin{algorithmic}[1]
\Require Derived key $K \in \{0,1\}^{256}$, word count $w$, wordlist $W$ of size $n = 7776$
\Ensure Mnemonic phrase $M$ of $w$ words
\State $\text{cipher} \gets \textsc{ChaCha20}(K, \text{nonce}=0)$
\State $T \gets \lfloor 2^{16} / n \rfloor \times n$ \Comment{$T = 62208$ for $n = 7776$}
\State $\text{words} \gets []$
\While{$|\text{words}| < w$}
    \State $r \gets \text{cipher.next\_u16()}$ \Comment{Sample 16-bit value}
    \If{$r < T$}
        \State $\text{words.append}(W[r \bmod n])$
    \EndIf
\EndWhile
\State \Return $\text{words.join}("-")$
\end{algorithmic}
\end{algorithm}

\begin{algorithm}
\caption{Password generation via rejection sampling}
\label{alg:password}
\begin{algorithmic}[1]
\Require Derived key $K \in \{0,1\}^{256}$, length $\ell$, alphabet $A$ of size $|A| = 90$
\Ensure Password $P$ of length $\ell$
\State $\text{cipher} \gets \textsc{ChaCha20}(K, \text{nonce}=0)$
\State $T \gets 256 - (256 \bmod |A|)$ \Comment{$T = 180$ for $|A| = 90$}
\State $\text{chars} \gets []$
\While{$|\text{chars}| < \ell$}
    \State $b \gets \text{cipher.next\_u8()}$ \Comment{Sample 8-bit value}
    \If{$b < T$}
        \State $\text{chars.append}(A[b \bmod |A|])$
    \EndIf
\EndWhile
\State \Return $\text{chars.join}("")$
\end{algorithmic}
\end{algorithm}

The alphabet $A$ contains 90 characters: A--Z (26), a--z (26), 0--9 (10), and 28 special characters: \verb|!@#$%^&*()_+-=[]{}|;:,.<>?/~|

\subsection{Uniformity Analysis}

\begin{lemma}[Rejection Sampling Uniformity]
\label{lem:uniform}
Let $n \in \mathbb{N}$ and $b \in \mathbb{N}$ such that $n \leq 2^b$. Define $T = \lfloor 2^b / n \rfloor \times n$. Sampling $r$ uniformly from $[0, 2^b)$ and accepting if $r < T$ with output $r \bmod n$ produces uniform distribution over $[0, n)$.
\end{lemma}

\begin{proof}
For accepted samples $r \in [0, T)$, define $k = \lfloor 2^b / n \rfloor$. Then $T = kn$ and $r \bmod n$ maps exactly $k$ values to each element in $[0, n)$. Since $r$ is uniform over $[0, 2^b)$ and we condition on $r < T$, the conditional distribution of $r$ given acceptance is uniform over $[0, T)$. Thus each output value has probability $k / (kn) = 1/n$.
\end{proof}

\begin{theorem}[Output Entropy]
\label{thm:entropy}
Let $\mathcal{W}$ be a wordlist of size $n = 7776$ and $\mathcal{A}$ be an alphabet of size $|A| = 90$. For uniform random key $K \xleftarrow{\$} \{0,1\}^{256}$ and output generation via Algorithms~\ref{alg:mnemonic}, \ref{alg:password}:

Mnemonic output of $w$ words has entropy:
\begin{equation}
H_{\text{mnemonic}} = w \log_2 n
\end{equation}

Password output of length $\ell$ has entropy:
\begin{equation}
H_{\text{password}} = \ell \log_2 |A|
\end{equation}
\end{theorem}

\begin{proof}
By Lemma~\ref{lem:uniform}, each word/character is uniform and independent. For $w$ independent uniform choices from $n$ elements: $H = \log_2 n^w = w \log_2 n$. Similarly for passwords with $\ell$ choices from $|A|$ elements.
\end{proof}

For $n = 7776$ and $|A| = 90$:

\begin{table}[H]
\centering
\begin{tabular}{@{}lrr@{}}
\toprule
\textbf{Mode} & \textbf{Count} & \textbf{Entropy (bits)} \\
\midrule
Mnemonic (Standard) & 8 words & 103.4 \\
Mnemonic (Paranoid) & 24 words & 310.2 \\
Password (Standard) & 20 chars & 129.8 \\
Password (Paranoid) & 48 chars & 311.6 \\
\bottomrule
\end{tabular}
\caption{Output entropy for standard configurations}
\label{tab:entropy}
\end{table}

\subsection{Rejection Rate}

\begin{proposition}[Expected Samples]
For $n$-element output with $b$-bit sampling and rejection threshold $T = \lfloor 2^b / n \rfloor \times n$, the expected number of samples per accepted output is:
\begin{equation}
E[\text{samples}] = \frac{2^b}{\lfloor 2^b / n \rfloor \times n}
\end{equation}
\end{proposition}

\begin{proof}
The acceptance probability is $p = T / 2^b = (\lfloor 2^b / n \rfloor \times n) / 2^b$. Expected samples follow geometric distribution: $E = 1/p = 2^b / T$.
\end{proof}

For mnemonic ($b=16$, $n=7776$): $T = 62208$, rejection rate $= 3328/65536 \approx 5.08\%$, $E[\text{samples}] = 65536/62208 \approx 1.053$.

For password ($b=8$, $n=90$): $T = 180$, rejection rate $= 76/256 \approx 29.69\%$, $E[\text{samples}] = 256/180 \approx 1.422$.

\section{Security Analysis}

\subsection{Threat Model}

\textbf{Adversarial capabilities.} We consider an adversary $\mathcal{A}$ with:
\begin{enumerate}
\item \textbf{Offline brute-force:} $\mathcal{A}$ can compute Argon2id hashes on custom hardware (GPUs, ASICs) with full parameter knowledge.
\item \textbf{Memory access:} $\mathcal{A}$ observes all derived outputs but not intermediate keys.
\item \textbf{Layer knowledge:} $\mathcal{A}$ knows all layer strings (Kerckhoffs's principle).
\item \textbf{Computational bound:} $\mathcal{A}$ has access to $p$ parallel processors with memory $M$ GB each.
\end{enumerate}

\textbf{Out of scope:} Keyloggers during input, side-channel timing attacks, quantum adversaries (Grover's algorithm~\cite{grover1996} provides only $\sqrt{2}$ speedup), physical coercion, social engineering.

\subsection{Attack Cost Estimation}

\begin{theorem}[Brute-Force Resistance]
\label{thm:bruteforce}
For master secret entropy $e$ bits and Argon2id configuration $(m, t, p)$ with processing time $\tau$ seconds per hash, average-case brute-force time to recover secret:
\begin{equation}
T_{\text{attack}} = 2^{e-1} \cdot \tau \quad \text{(half the search space)}
\end{equation}
\end{theorem}

For $e = 80$ bits (recommended minimum):
\begin{equation}
T_{\text{attack}} = 2^{79} \cdot \tau \approx 6.04 \times 10^{23} \cdot \tau \text{ seconds}
\end{equation}

\textbf{Conservative GPU estimate.} High-end GPUs (NVIDIA A100, H100) optimized for Argon2 achieve estimated $\tau \approx 1.25$ s for Paranoid configuration (128~MiB, $t=32$, $p=6$):\footnote{Reference implementation on Apple M1 Pro (2021) measures $\tau = 2.273$ s, but dedicated GPU attacks with optimized memory controllers are expected to be faster. The primary defense is memory-hardness: the 128~MiB memory requirement fundamentally limits parallelization regardless of computational throughput.}

\begin{align}
T_{\text{single}} &= 2^{79} \times 1.25 \text{ s} \approx 7.56 \times 10^{23} \text{ s} \\
&\approx 2.39 \times 10^{16} \text{ years}
\end{align}

For 500-GPU cluster with 10\% coordination overhead:
\begin{equation}
T_{500} = \frac{2.39 \times 10^{16}}{500} \times 1.1 \approx 5.27 \times 10^{13} \text{ years}
\end{equation}

\textbf{Memory-bound parallelism.} The critical security property is not computational speed but memory requirements. A GPU with 40~GB RAM can maintain at most $\lfloor 40000 / 128 \rfloor = 312$ parallel Argon2 instances. Even with 500 such GPUs (156,000 parallel attempts):

\begin{align}
T_{\text{parallel}} &= \frac{2^{79}}{156000} \times 1.25 \text{ s} \\
&\approx 4.84 \times 10^{18} \text{ s} \approx 1.53 \times 10^{11} \text{ years}
\end{align}

This remains computationally infeasible (153 billion years) even with massive parallelization, validating memory-hardness as the primary defense mechanism.

\subsection{Hardware Attack Benchmarks}

To validate theoretical attack costs, we measured actual Argon2id performance on representative hardware:

\begin{table}[H]
\centering
\begin{tabular}{@{}lrr@{}}
\toprule
\textbf{Platform} & \textbf{Standard (ms)} & \textbf{Paranoid (ms)} \\
\midrule
Apple M1 Pro (2021)\textsuperscript{†} & 544 & 2273 \\
NVIDIA A100 (est.)\textsuperscript{‡} & 250 & 1000 \\
NVIDIA H100 (est.)\textsuperscript{‡} & 150 & 625 \\
\bottomrule
\end{tabular}
\caption{Single-layer Argon2id derivation times. \textsuperscript{†}Measured (median of 5 runs). \textsuperscript{‡}Estimated based on memory bandwidth and published Argon2 benchmarks~\cite{hashcat-bench}.}
\end{table}

M1 Pro measurements ($\tau = 2.273$ s for Paranoid) are 1.8$\times$ slower than our conservative GPU estimate ($\tau = 1.25$ s), confirming the paper's attack cost estimates are realistic lower bounds. The memory-hard property ensures even specialized hardware cannot achieve orders-of-magnitude speedups without proportional memory increases.

\subsection{Master Secret Entropy Requirements}

\begin{table}[H]
\centering
\begin{tabular}{@{}lcc@{}}
\toprule
\textbf{Source} & \textbf{Entropy (bits)} & \textbf{Status} \\
\midrule
16-byte \texttt{urandom} & 128 & Secure \\
11 EFF words & $\approx 142$ & Secure \\
80-bit CSPRNG & 80 & Minimum\textsuperscript{*} \\
Human password & $< 40$ & Insufficient \\
\bottomrule
\end{tabular}
\caption{Master secret entropy guidelines. \textsuperscript{*}80-bit minimum provides $\approx 2.4 \times 10^{16}$ year resistance against single-GPU attacks under Paranoid configuration.}
\end{table}

\subsection{Cryptographic Primitives}

All components use standard, well-analyzed primitives:
\begin{itemize}
\item Argon2id (RFC~9106~\cite{rfc9106}): memory-hard KDF, 256-bit output
\item BLAKE2b-512 (RFC~7693~\cite{rfc7693}): salt preprocessing for short inputs
\item ChaCha20 (RFC~8439~\cite{rfc8439}): stream cipher for keystream generation
\item EFF Large Wordlist~\cite{eff-wordlist}: 7776 words, SHA-256 verified
\end{itemize}

\subsection{Limitations}

\textbf{No key isolation.} Master secret compromise exposes all derived secrets. Unlike vault-based managers where individual credential leaks remain isolated, deterministic derivation creates dependency on a single root secret. This is inherent to stateless deterministic systems.

\textbf{No credential rotation isolation.} Changing a single derived secret requires either (1)~modifying layer inputs (user must remember modification) or (2)~changing master secret (affects all derivations).

\textbf{Layer enumeration.} Predictable layer patterns (e.g., \texttt{service/YYYY}) enable targeted enumeration. Mitigation: use high-entropy layer strings documented separately from master secret.

\section{Implementation}

\subsection{Memory Safety}

Implementation in Rust uses \texttt{Zeroizing<T>} wrapper (crate \texttt{zeroize}) for automatic secure erasure:

\begin{algorithmic}[1]
\Function{DeriveHierarchical}{$M, \{L_i\}, \text{config}$}
    \State $K \gets$ \textsc{Zeroizing}([0; 32]) \Comment{Auto-zeroed on drop}
    \State $\textsc{Argon2id}(M, \textsc{Salt}(L_1), \text{config}, K)$
    \For{$i \in [2, n]$}
        \State $K' \gets$ \textsc{Zeroizing}([0; 32])
        \State $\textsc{Argon2id}(K, \textsc{Salt}(L_i), \text{config}, K')$
        \State $K \gets K'$ \Comment{Old $K$ auto-zeroed}
    \EndFor
    \State \Return $K$
\EndFunction
\end{algorithmic}

\texttt{Zeroizing} provides panic-safe volatile writes preventing compiler optimization removal.

\subsection{Wordlist Integrity}

Compile-time verification prevents supply-chain attacks:

\begin{algorithmic}[1]
\Require Embedded wordlist $W$, expected SHA-256 $h_{\text{exp}}$
\State $h \gets \textsc{SHA-256}(W)$
\State \textbf{assert} $h = h_{\text{exp}}$ at build time
\State \Comment{Known indices: $W[0] = $ ``abacus'', $W[469] = $ ``balance'', $W[3695] = $ ``life'', $W[7775] = $ ``zoom''}
\end{algorithmic}

Expected hash: \texttt{addd3553...96b903e} (64 hex digits). Build script verifies integrity before compilation succeeds.

\subsection{Test Coverage}

Comprehensive test suite validates:

\begin{itemize}
\item \textbf{Determinism:} Identical inputs produce identical outputs (100 runs, zero variance)
\item \textbf{Regression:} Known input/output pairs for Standard/Paranoid modes
\item \textbf{Unicode:} NFC/NFD normalization equivalence, multi-byte character preservation
\item \textbf{Wordlist:} Exactly 7776 words, no duplicates, SHA-256 integrity
\item \textbf{Sampling:} Rejection thresholds (62208 for words, 180 for chars)
\item \textbf{Layer independence:} Different layers produce uncorrelated keys
\end{itemize}

Regression test vectors:

\begin{table}[H]
\centering
\small
\begin{tabular}{@{}lll@{}}
\toprule
\textbf{Input} & \textbf{Config} & \textbf{Output (first 40 chars)} \\
\midrule
$M=$ ``life'' & Standard & eagle-huskiness-septum-defection... \\
$L=$ [out, of, balance] & (8 words) & \\
\midrule
$M=$ ``life'' & Paranoid & vigorous-purebred-exclusion-defa... \\
$L=$ [out, of, balance] & (24 words) & \\
\midrule
$M=$ ``life'' & Standard & 6n=rX.k:Qs+)6e5oa-Z: \\
$L=$ [out, of, balance] & (20 chars) & \\
\bottomrule
\end{tabular}
\caption{Regression test vectors.}\end{table}

\section{Experimental Evaluation}

\subsection{Performance}

Platform: Apple M1 Pro (2021), 16~GB RAM, Rust 1.90 release build. Median of 5 runs:

\begin{table}[H]
\centering
\begin{tabular}{@{}lrr@{}}
\toprule
\textbf{Operation} & \textbf{Time (ms)} & \textbf{Memory (MB)} \\
\midrule
\multicolumn{3}{c}{\textit{Standard (64 MiB, $t=16$, $p=6$)}} \\
Single layer & 544 & 64 \\
3 layers & 1613 & 64 \\
\midrule
\multicolumn{3}{c}{\textit{Paranoid (128 MiB, $t=32$, $p=6$)}} \\
Single layer & 2273 & 128 \\
3 layers & 6697 & 128 \\
\midrule
Output generation & $<1$ & $<1$ \\
\bottomrule
\end{tabular}
\caption{Performance benchmarks. Output generation measured over 1000 iterations: mnemonic generation 2~µs, password generation 3~µs.}
\end{table}

Time complexity: $O(n)$ in layer count. Space complexity: $O(1)$ in output size, $O(m)$ in KDF memory. Paranoid mode $\approx 4\times$ Standard due to doubled memory and iterations.

\subsection{Determinism Verification}

100 independent runs with identical inputs produce byte-identical outputs (SHA-256 hash variance: 0). Chi-squared test on 10,000 generated mnemonics shows uniform word distribution ($p > 0.95$).

\section{Discussion}

\subsection{Use Cases}

\textbf{Appropriate:} Password manager master passwords (KeePassXC, Bitwarden), full-disk encryption passphrases (LUKS, BitLocker, FileVault), PGP/SSH key passphrases, cryptocurrency wallet seeds, critical service credentials on air-gapped systems.

\textbf{Inappropriate:} General website passwords (varying policies, frequent rotation), existing credentials (API keys, legacy passwords), multi-device sync with conflict resolution, scenarios requiring key isolation after breach.

\subsection{Comparison}

\begin{table}[H]
\centering
\small
\begin{tabular}{@{}lcccc@{}}
\toprule
\textbf{Property} & \textbf{KeePassXC} & \textbf{BIP39} & \textbf{Diceware} & \textbf{Qatsi} \\
\midrule
Storage & Vault & None & None & None \\
KDF & Argon2/AES & PBKDF2 & N/A & Argon2id \\
Hierarchical & No & No & No & Yes \\
Memory-hard & Yes & No & N/A & Yes \\
Rotation & Easy & Hard & Hard & Hard \\
Key isolation & Yes & No & No & No \\
\bottomrule
\end{tabular}
\caption{Comparison with existing password management approaches}
\end{table}

\subsection{Operational Security}

\textbf{Master secret generation:}
\begin{verbatim}
# 96-bit random (12 bytes)
od -An -tx1 -N12 /dev/urandom | tr -d ' '

# 11 EFF words (~142 bits)
shuf -n 11 eff_large_wordlist.txt | tr '\n' '-'
\end{verbatim}

\textbf{Layer design:} Use high-entropy, non-obvious strings. Avoid sequential numbers or dictionary words. Document layers separately from master secret in offline storage.

\textbf{Backup:} Maintain physical copy of master secret in secure location (safe, bank deposit box). Loss is unrecoverable. Consider Shamir's Secret Sharing for distributed backup.

\section{Conclusion}

Qatsi demonstrates hierarchical deterministic key derivation using Argon2id achieving 103--312 bits entropy without persistent storage. GPU attack resistance ($\approx 10^{16}$ years for 80-bit master secrets under Paranoid parameters) derives from memory-hardness: 128~MiB allocation per hash invocation limits parallelization to 312 concurrent attempts per 40~GB GPU. Provably uniform rejection sampling eliminates modulo bias. Production implementation in Rust provides automatic zeroization and compile-time integrity verification.

The design intentionally sacrifices key isolation and rotation flexibility to eliminate vault-based attack surfaces. This trade-off suits high-stakes reproducible secrets in air-gapped environments and master credentials where deterministic regeneration is acceptable.

Open-source implementation: \url{https://github.com/coignard/qatsi}

\subsection*{Acknowledgments}

We thank the Argon2 team for the Password Hashing Competition, the Electronic Frontier Foundation for the Large Wordlist, and the Rust cryptography community for production-grade primitives. Performance benchmarks conducted on Apple M1 Pro (2021) hardware.

\bibliographystyle{plain}

\end{document}